\documentclass [10pt]{article}

\usepackage[utf8]{inputenc}
\usepackage{fullpage}
\usepackage{amsmath, cite, amssymb, graphicx}

\newtheorem{lemma}{Lemma}
\newtheorem{theorem}{Theorem}

\newcommand{\qed}{\hfill\ensuremath{\Box}\medskip\\\noindent}
\newenvironment{proof}{\noindent\emph{Proof. }}{}

\newcommand{\depth}{\ensuremath{\mathrm{depth}}}

\newcommand{\LA}{\ensuremath{\mathrm{LA}}}
\newcommand{\NCA}{\ensuremath{\mathrm{NCA}}}
\newcommand{\LCE}{\ensuremath{\mathrm{LCE}}}
\newcommand{\LCEPP}{\ensuremath{\mathrm{LCE}_{\mathit{PP}}}}
\newcommand{\LCEPT}{\ensuremath{\mathrm{LCE}_{\mathit{PT}}}}
\newcommand{\LCETT}{\ensuremath{\mathrm{LCE}_{\mathit{TT}}}}

\newcommand{\polylog}{\text{\hspace{1mm}polylog}}

\usepackage[dvipsnames,usenames]{color}
\usepackage[colorlinks=true,urlcolor=Blue,citecolor=Green,linkcolor=BrickRed]{hyperref}
\begin{document}
\title{Longest Common Extensions in Trees\footnote{A preliminary version of this paper appeared in the Proceedings of the 26th Annual symposium on Combinatorial Pattern Matching, 2015.}}
\author{Philip Bille\thanks{Partially supported by the Danish Research Council (DFF – 4005-00267, DFF – 1323-00178) and the Advanced Technology Foundation.} \\  \href{mailto:phbi@dtu.dk}{phbi@dtu.dk} \and Pawe{\l} Gawrychowski\thanks{Work done while the author held a post-doctoral position
at the Warsaw Center of Mathematics and Computer Science.}\\ \href{mailto:gawry@mimuw.edu.pl}{gawry@mimuw.edu.pl} \and Inge Li G{\o}rtz$^{\ast}$\\
 \href{mailto:inge@dtu.dk}{inge@dtu.dk}    \and  
 Gad M. Landau\thanks{Partially supported by the National Science Foundation
Award 0904246, Israel Science Foundation grant  571/14,
Grant No. 2008217 from the United States-Israel
Binational Science Foundation (BSF) and DFG.} \\  \href{mailto:landau@cs.haifa.ac.il}{landau@cs.haifa.ac.il}   
\and Oren Weimann$^{\ast}$\thanks{Partially supported by  the Israel Science Foundation grant 794/13.} \\  \href{mailto:oren@cs.haifa.ac.il}{oren@cs.haifa.ac.il} }
\date{}

\maketitle

\begin{abstract}

The longest common extension (LCE) of two indices in a string is the length of the longest identical substrings starting at these two indices. The LCE problem asks to preprocess a string into a compact data structure that supports fast LCE queries.
 
In this paper we generalize the LCE problem to trees and suggest a few applications of LCE in trees to tries and XML databases. Given a labeled and rooted tree $T$ of size $n$, the goal is to preprocess $T$ into a compact data structure that support the following LCE queries between subpaths and subtrees in $T$. Let  $v_1$, $v_2$, $w_1$, and $w_2$ be nodes of $T$ such that $w_1$ and $w_2$ are descendants of $v_1$ and $v_2$ respectively.

\begin{itemize}
\item $\LCEPP(v_1, w_1, v_2, w_2)$: (path-path $\LCE$) return the longest common  prefix of the paths $v_1 \leadsto w_1$ and $v_2 \leadsto w_2$.
\item $\LCEPT(v_1, w_1, v_2)$: (path-tree $\LCE$) return maximal path-path LCE of the path $v_1 \leadsto w_1$ and any path from $v_2$ to a descendant leaf.
\item $\LCETT(v_1, v_2)$: (tree-tree $\LCE$) return a maximal path-path LCE of any pair of paths from $v_1$ and $v_2$ to  descendant leaves.
\end{itemize}
We present the first non-trivial bounds for supporting these queries. For $\LCEPP$ queries, we present a linear-space solution with $O(\log^{*} n)$ query time. For $\LCEPT$  queries, we present a linear-space solution with $O((\log\log n)^{2})$ query time, and complement this with a lower bound showing that any path-tree LCE structure of size $O(n \polylog(n))$ must necessarily
use $\Omega(\log\log n)$ time to answer queries.
For $\LCETT$ queries, we present a time-space trade-off, that given any parameter $\tau$, $1 \leq \tau \leq n$, leads to an $O(n\tau)$ space and $O(n/\tau)$ query-time solution (all of these bounds hold on a standard unit-cost RAM model with logarithmic word size). This is complemented with a reduction from the set intersection problem implying that a fast linear space solution is not likely to exist.  

\end{abstract}

\section{Introduction}
Given a string $S$, the \emph{longest common extension} (LCE) of two indices is the length of the longest identical substring starting at these indices. The \emph{longest common extension problem} (LCE problem) is to preprocess $S$ into a compact data structure supporting fast LCE queries. The LCE problem is a well-studied basic primitive~\cite{BGK2012,INL2010,BGSV2014, BGKLV2015, FH2006} with a wide range of applications in problems  such as approximate string matching, finding exact and approximate tandem repeats, and  finding palindromes~\cite{ALP2004,CH2002, LV1989, Myers1986,GS2004,LSS2001, LMS1998, ML1984}. The classic textbook solution to the LCE problem on strings combines a suffix tree with a nearest common ancestor (NCA) data structure leading to a linear space and constant query-time solution~\cite{Gusfield1997}.

In this paper we study generalizations of the LCE problem to trees. The goal is to preprocess an edge-labeled, rooted tree $T$ to support the various $\LCE$ queries between paths in $T$. Here a path starts at a node $v$ and ends at a descendant of $v$, and the LCEs are on the strings obtained by concatenating the characters on the edges of the path from top to bottom (each edge contains a single character). We consider path-path LCE queries between two specified paths in $T$, path-tree LCE queries defined as the maximal path-path LCE of a given path and \emph{any} path starting at a given node, and tree-tree LCE queries defined as the maximal path-path LCE between \emph{any} pair of paths starting from two given nodes.   We next define these problems formally. 
%

\paragraph{\bf Tree LCE Problems.}
Let $T$ be an edge-labeled, rooted tree with $n$ nodes. 
We denote the subtree rooted at a node $v$ by $T(v)$, and  given nodes $v$ and $w$ such that $w$ is in $T(v)$ the path going down from $v$ to $w$ is denoted $v \leadsto w$. A \emph{path prefix} of $v \leadsto w$ is any subpath $v \leadsto u$ such that $u$ is on the path $v \leadsto w$. Two paths $v_1 \leadsto w_1$ and $v_2 \leadsto w_2$ \emph{match} if  concatenating the labels of all edges in the paths gives the same string. Given nodes  $v_1, w_1$ such that $w_1 \in T(v_1)$ and nodes $v_2, w_2$ such that  $w_2 \in T(v_2)$ define the following queries:
\begin{itemize}
\item $\LCEPP(v_1, w_1, v_2, w_2)$: (path-path $\LCE$) return the longest common matching prefix of the paths $v_1 \leadsto w_1$ and $v_2 \leadsto w_2$.
\item $\LCEPT(v_1, w_1, v_2)$: (path-tree $\LCE$) return the maximal path-path LCE of the path $v_1 \leadsto w_1$ and any path from $v_2$ to a descendant leaf.
\item $\LCETT(v_1, v_2)$: (tree-tree $\LCE$) return a maximal path-path LCE of any pair of paths from $v_1$ and $v_2$ to  descendant leaves.
\end{itemize}
The queries are illustrated in Fig.~\ref{fig:LCEexample}. We assume that the output of the queries is reported compactly as the endpoint(s) of the $\LCE$. This allows us to report the shared path in constant time. Furthermore, we will assume w.l.o.g. that for each node $v$ in $T$, all the edge-labels to children of  $v$ are distinct. If this is not the case, then we can merge all identical edges of a node to its children in linear time, without affecting the result of all the above $\LCE$ queries. 

\begin{figure}[t] 
   \centering
   \includegraphics[scale=.55]{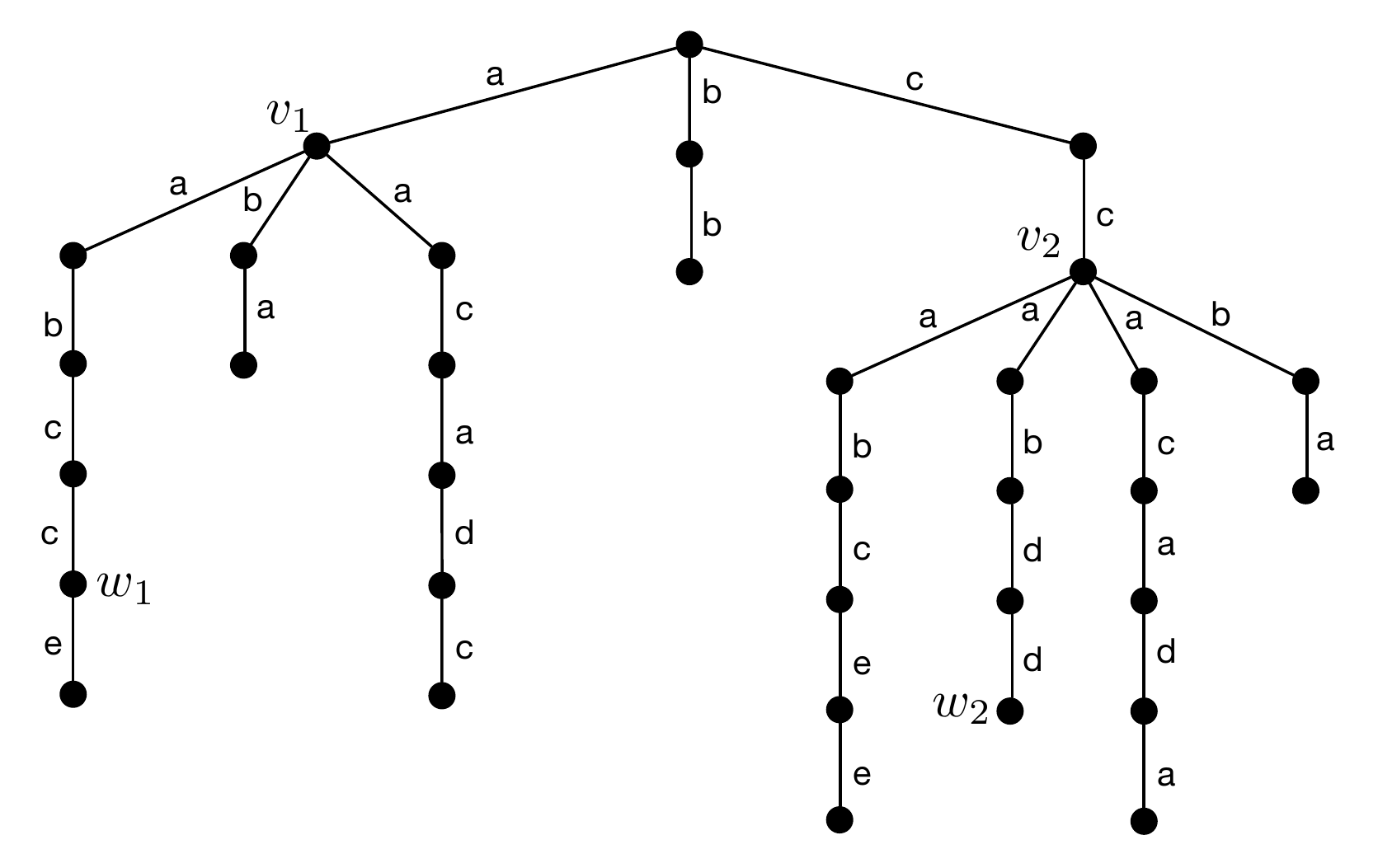} 
   \caption{LCE in trees. $\LCEPP(v_1, w_1, v_2, w_2)$ is the path ''ab'', $\LCEPT(v_1, w_1, v_2)$ is the path ''abc'', and $\LCETT(v_1, v_2)$ is the path ''acad''.}
   \label{fig:LCEexample}
\end{figure}
We note that the direction of the paths in $T$ is important for the LCE queries. In the above LCE queries, the paths start from a node and go downwards. If we instead consider paths from a node going upwards towards the root of $T$, the problem is easier and can be solved in linear space and constant query-time by combining Breslauer's suffix tree of a tree~\cite{Breslauer1998} with a nearest common ancestor (NCA) data structure~\cite{HT1984}.


\paragraph{\bf Our Results.}
First consider the $\LCEPP$ and $\LCEPT$ problems. To answer an $\LCEPP(v_1, w_1, v_2, w_2)$ query, a straightforward solution is to traverse both paths in parallel-top down. Similarly, to answer an $\LCEPT(v_1, w_1, v_2)$ query we can traverse $v_1 \leadsto w_1$  top-down while traversing the matching path from $v_2$ (recall that all edges to a child are distinct and hence the longest matching path is unique). This approach leads to a linear-space  solution with $O(h)$ query-time  to both problems, where $h$ is the height of $T$. Note that for worst-case trees we have that $h = \Omega(n)$. 

We show the following results. For $\LCEPP$ we give a linear $O(n)$ space and $O(\log^{*} n)$ query-time solution. For  $\LCEPT$  we give a linear  $O(n)$ space and $O((\log\log n)^{2})$ query-time solution, and complement this with a lower bound stating that any  $\LCEPT$ solution using $O(n \polylog(n))$ space must necessarily
have $\Omega(\log\log n)$ query time.
 

Next consider the $\LCETT$ problem. Here, the simple top down traversal does not  work and it seems that substantially different ideas are needed. We first show a reduction from the \emph{set-intersection problem}, i.e., preprocessing a family of sets of total size $n$ to support disjointness queries between any pairs of sets. In particular, the reduction implies that a fast linear space solution is not likely assuming a widely believed conjecture on the complexity of the set-intersection problem. We complement this result with a time-space trade-off that achieves $O(n\tau)$ space and $O(n/\tau)$ query time for any parameter $1 \leq \tau \leq n$.

All results assume the standard word RAM model with word size $\Theta(\log n)$. We also assume the alphabet is either sorted or is linear-time sortable.

\paragraph{\bf Applications.}
We suggest a few immediate applications of LCE in trees. Consider a set of strings $\mathcal{S} = \{S_1, \ldots, S_k\}$ of total length $\sum_{i = 1}^k |S_i| = N$ and let $T$ be the \emph{trie} of $\mathcal{S}$ of size $n$, i.e., $T$ is the labeled, rooted tree obtained by merging shared prefixes in $\mathcal{S}$ maximally. If we want to support $\LCE$ queries between suffixes of strings in $\mathcal{S}$, the standard approach is to build a generalized suffix tree for the strings and combine it with an NCA data structure. This leads to a solution using $O(N)$ space and $O(1)$ query time. We can instead implement the $\LCE$ query between the suffixes of strings in $\mathcal{S}$ as an $\LCEPP$ on the trie $T$, i.e., any suffix of a string in $\mathcal{S}$ corresponds to a path in $T$ and hence the $\LCE$ of two suffixes in $\mathcal{S}$ corresponds to an $\LCEPP$ query in $T$. With our data structure for  $\LCEPP$, this leads to  a solution using $O(n)$ space and $O(\log^{*} n)$ query time. In general, $n$ can be significantly smaller than $N$, depending on the amount of shared prefixes in $\mathcal{S}$.  Hence, this solution provides a more space-efficient representation of $\mathcal{S}$ at the expense of a tiny increase in query time. An $\LCEPT$ query on $T$ corresponds to computing a maximal $\LCE$ of a suffix of a string in $\mathcal{S}$ with suffixes of strings in $\mathcal{S}$ sharing a common prefix. An $\LCETT$ query on $T$ corresponds to computing a maximal $\LCE$ over pairs of suffixes of strings in $\mathcal{S}$ that share a common prefix. To the best of our knowledge these queries are novel one-to-many and many-to-many $\LCE$ queries. Since tries are a basic data structure for storing strings we expect these queries to be of interest in a number of applications.

Another interesting application is using $\LCE$ in trees as a query primitive for XML data. XML documents can be viewed as a labeled tree and typical queries (e.g., XPath queries) involve traversing and identifying paths in the tree. The $\LCE$ queries provide simple and natural primitives for comparing paths and subtrees without explicit traversal. For instance, our  solution for $\LCEPT$ queries can be used to quickly identify the ``best match'' of a given path in a subtree.

\section{Preliminaries} 

Given a node $v$ and an integer $d \ge 0$, the \emph{level ancestor} of $v$ at depth $d$, denoted $\LA(v, d)$ is the ancestor of $v$ at depth $d$.
We explicitly compute and store the depth of every node $v$, denoted $\depth(v)$.
Given a pair of nodes $v$ and $w$ the \emph{nearest common ancestor} of $v$ and $w$, denoted $\NCA(v, w)$, is the common ancestor of $v$ and $w$ of greatest depth. Both $\LA$ and $\NCA$ queries can be supported in constant time with a linear space data structures, see e.g.,~\cite{BV1994, Dietz1991, BFC2004, AH2000, GRV2006, AGKR2004,HT1984, BFC2000, FH2011}

Finally, the \emph{suffix tree of a tree}~\cite{Kosaraju1989, Breslauer1998, shibuya1999} is the compressed trie of all suffixes of leaf-to-root paths in $T$. The suffix tree of a tree uses $O(n)$ space and  can be constructed  in $O(n\log \log n)$ time for general alphabets~\cite{shibuya1999}. Note that the suffix tree of a tree combined with \NCA\ can support \LCE\ queries in constant time for paths going upwards. Since we consider paths going downwards, we will only use the suffix tree to check (in constant time) if two paths are completely identical.

We also need the following three primitives. {\em Range minimum queries:} A list of $n$ numbers $a_{1},a_{2},\ldots a_{n}$ can be augmented with $2n+o(n)$ bits of additional
data in $O(n)$ time, so that for any $i\leq j$ the position of the smallest number among $a_{i},a_{i+1},\ldots,a_{j}$ can be found in $O(1)$
time~\cite{FH2011}.  {\em Predecessor queries:} Given a sorted collection of $n$ integers from $[0,U)$, a structure of size $O(n)$ answering predecessor queries in
$O(\log\log U)$ time can be constructed in time $O(n)$~\cite{BKZ1977}, where a predecessor query locates, for a given $x$, the largest $y\leq x$
such that $y\in S$. Finally,  {\em Perfect hashing:}  given a collection $S$ of $n$ integers a perfect hash table can be constructed in expected $O(n)$ time~\cite{FKS},
where a perfect hash table checks, for a given $x$, if $x\in S$, and if so returns its associated data in $O(1)$ time. The last result can be
made deterministic at the expense of increasing the preprocessing time to $O(n\log\log n)$~\cite{Ruzic2004}, but then we need one additional step
in our solution for the path-tree LCE as to ensure $O(n)$ total construction time.

\section{Difference Covers for Trees}
In this section we introduce a generalization of difference covers from strings to trees. This will be used to decrease the space of our data structures. We believe it is of independent interest.

\begin{lemma}\label{lem:difference}
For any tree $T$ with $n$ nodes and a parameter $x$, it is possible to mark $2n/x$ nodes of $T$, so that for any two nodes $u,v\in T$
at (possibly different) depths at least $x^{2}$, there exists $d\leq x^{2}$ such that the $d$-th ancestors of both $u$ and $v$
are marked. Furthermore, such $d$ can be calculated in $O(1)$ time and the set of marked nodes can be determined in $O(n)$ time.
\end{lemma}

\begin{proof}
We distinguish between two types of marked nodes. Whether a node $v$ is marked or not depends only on its depth.
The marked nodes are determined as follows.

\begin{description}
\item[Type I.] For every $i=0,1,\ldots,x-1$, let $V_{i}$ be the set of nodes at depth leaving a remainder of $i$ when divided by $x$. Because
$\bigcup_{i}V_{i}=T$ and all $V_{i}'s$ are disjoint, there exists $r_{1}\in [0,x-1]$ such that $|V_{r_{1}}|\leq n/x$. Then $v$ is a type I marked node
iff $\depth(v) = r_{1}  \bmod x$.
\item[Type II.]  For every $i=0,1,\ldots,x-1$, let $V_{i}$ be the set of nodes $v$ such that $\lfloor \depth(v)/x \rfloor$ leaves a remainder
of $i$ when divided by $x$. By the same argument as above, there exists $r_{2}\in [0,x-1]$ such that $|V_{r_{2}}|\leq n/x$. Then $v$ is a type II marked
node iff $\lfloor \depth(v)/x \rfloor = r_{2}  \bmod x$.
\end{description}

Now, given two nodes $u$ and $v$ at depths at least $x^{2}$, we need to show that there exists an appropriate $d\leq x^{2}$. Let
$\depth(u)=t_{1} \bmod x$ and choose $d_{1}=t_{1}+x-r_{1}$. Then the $d_{1}$-th ancestor of $u$ is a type I marked node, because
its depth is $\depth(u)-d_{1}=\depth(u)-(t_{1}+x-r_{1})=\depth(u)-t_{1}-x+r_{1}$, which leaves a remainder of $r_{1}$ when divided by $x$.
Our $d$ will be of the form $d_{1}+d_{2}x$. Observe that regardless of the value of $d_{2}$, we can be sure that the $d$-th ancestor
of $u$ is a type I marked node. Let $v'$ be the $d_{1}$-th ancestor of $v$, $\lfloor \depth(v')/x \rfloor = t_{2} \bmod x$ and choose
$d_{2}=t_{2}+x-r_{2}$. The $(d_{2}x)$-th ancestor of $v'$ is a type II marked node, because
$\lfloor (\depth(v')-d_{2}x)/x \rfloor = \lfloor \depth(v')/x\rfloor -t_{2}-x+r_{2}$, which leaves a remainder of $r_{2}$ when divided by $x$.
Therefore, choosing $d=d_{1}+d_{2}x$ guarantees that $d\leq x-1+x(x-1)<x^{2}$, so the $d$-th ancestors of $u$ and $v$ are both
defined, the $d$-th ancestor of $u$ is a type I marked node, and the $d$-th ancestor of $v$ is a type II marked node.

The total number of marked nodes is clearly at most $2n/x$, and the values of $r$ and $r'$ can be determined by a single traversal of $T$.
To determine $d$, we only need to additionally know $\depth(u)$ and $\depth(v)$ and perform a few simple arithmetical operations.
\qed
\end{proof}

\noindent {\bf Remark.} Our difference cover has the following useful property: whether a node $v$ is marked or not depends only on
the value of $\depth(v) \pmod{x^{2}}$. Hence, if a node at depth at least $x^{2}$ is marked then so is its $(x^{2})$-th ancestor. Similarly, if
a node is marked, so are all of its descendants at distance $x^{2}$.

\section{Path-Path LCE} 
\label{sec:path-path}

In this section we prove the following theorem.

\begin{theorem}\label{thm:lcepp}
For a tree $T$ with $n$ nodes, a data structure of size $O(n)$ can be constructed in $O(n)$ time to answer path-path LCE queries in
$O(\log^{*} n)$ time. 
\end{theorem}

We start with a simple preliminary $O(n\log n)$-space $O(1)$-query data structure which will serve as a starting point for the more complicated final implementation. We note that a data structure with similar guarantees to Lemma~\ref{lem:simple path-path} is also implied from~\cite{Bannai2013}.

\begin{lemma}
For a tree $T$ with $n$ nodes, a data structure of size $O(n\log n)$ can be constructed in $O(n\log n)$ time to answer path-path LCE queries in
$O(1)$ time.
\label{lem:simple path-path}
\end{lemma}

\begin{proof}
The structure consists of $\log n$ separate parts, each of size $O(n)$. The $k$-th part answers in $O(1)$ time path-path LCE queries such that both
paths are of the same length $2^{k}$. This is enough to answer a general path-path LCE query in the same time complexity, because we can first
truncate the longer path so that both paths are of the same length $\ell$, then calculate $k$ such that $2^{k}\leq \ell < 2^{k+1}$. Then
we have two cases.

\begin{enumerate}
\item The prefixes of length $2^{k}$ of both paths are different. Then replacing the paths by their prefixes of length $2^{k}$ does not
change the answer.
\item The prefixes of length $2^{k}$ of both paths are the same. Then replacing the paths by their suffixes of length $2^{k}$ does not
change the answer.
\end{enumerate}

We can check if the prefixes are the same and then (with level ancestor queries) reduce the query so that both paths are
of the same length $2^{k}$, all in $O(1)$ time.

Consider all paths of length $2^{k}$ in the tree. There are at most $n$ of them, because every node $u$ creates at most one new path
$\LA(v,\depth(v)-2^{k}) \leadsto v$. We lexicographically sort all such paths and store the longest common extension of every two neighbours
on the sorted list. Additionally, we augment the longest common extensiones with a range minimum query structure, and keep at every $v$ the
position of the path $\LA(v,\depth(v)-2^{k}) \leadsto v$ (if any) on the sorted list. This allows us to answer
$\LCEPP(\LA(u,\depth(u)-2^{k}),u,\LA(v,\depth(v)-2^{k}),v)$ in $O(1)$ time: we lookup the positions of $\LA(u,\depth(u)-2^{k}) \leadsto u$
and $\LA(v,\depth(v)-2^{k}) \leadsto v$ on the sorted list and use the range minimum query structure to calculate their longest common
prefix, all in $O(1)$ time. The total space usage is $O(n)$, because every node stores one number and additionally we have a list of at most
$n$ numbers augmented with a range minimum query structure.

To construct the structure efficiently, we need to lexicographically sort all paths of length $k$. This can be done in $O(n)$ time for every $k$ after
observing that every path of length $2^{k+1}$ can be conceptually divided into two paths of length $2^{k}$. Therefore, if we have already
lexicographically sorted all paths of length $2^{k}$, we can lexicographically sort all paths of length $2^{k+1}$ by sorting pairs of numbers
from $[1,n]$, which are the positions of the prefix and the suffix of a longer path on the sorted list of all paths of length $2^{k}$. With
radix sorting, this takes $O(n)$ time. Then we need to compute the longest common extension of ever two neighbours on the sorted list,
which can be done in $O(1)$ time by using the already constructed structure for paths of length $2^{k}$. Consequently, the total construction
time is $O(n\log n)$.
\qed
\end{proof}

To decrease the space usage of the structure from Lemma~\ref{lem:simple path-path}, we use the difference covers developed in
Lemma~\ref{lem:difference}. Intuitively, the first step is to apply the lemma with $x=\log n$ and preprocess only paths of length $2^{k}\log^{2} n$
ending at the marked nodes. Because we have only $O(n/\log n)$ marked nodes, this requires $O(n)$ space. Then, given two paths of length
$\ell$, we can either immediately return their LCE using the preprocessed data, or reduce the query to computing
the LCE of two paths of length at most $\log^{2}n$. Using the same reasoning again with $x=\log (\log^{2}n)$,
we can reduce the length even further to at most $\log^{2}(\log^{2}n)$ and so on. After $O(\log^{*}n)$ such reduction steps, we guarantee that the
paths are of length $O(1)$, and the answer can be found naively. Formally, every step is implemented using the following lemma.

\begin{lemma}
For a tree $T$ with $n$ nodes and a parameter $b$, a data structure of size $O(n)$ can be constructed in $O(n)$ time, so that given two paths of
length at most $b$ ending at $u\in T$ and $v\in T$ in $O(1)$ time we can either compute the path-path LCE or reduce the query so that
the paths are of length at most $\log^{2}b$.
\label{lem:path-path reduction}
\end{lemma}

\begin{proof}
We apply Lemma~\ref{lem:difference} with $x=\log b$. Then, for every $k=0,1,\ldots,\log b$ separately, we consider all paths of length
$2^{k}\log^{2} b$ ending at marked nodes. As in the proof of Lemma~\ref{lem:simple path-path}, we lexicographically sort all such paths
and store the longest common extension of every two neighbours on the sorted list augmented with a range minimum query structure. Because
we have only $O(n/\log b)$ marked nodes, the space decreases to $O(n)$. Furthermore, because the length of the paths is of the form
$2^{k}\log^{2}b$ (as opposed to the more natural choice of $2^{k}$), all lists can be constructed in $O(n)$ total time by radix sorting, as
a path of length $2^{k+1}\log^{2}b$ ending at a marked node can be decomposed into two paths of length $2^{k}\log^{b}$ ending
at marked nodes, because if a node is marked, so is its $(x^{2})$-th ancestor.

Consider two paths of the same length $\ell \leq b$ ending at $u\in T$ and $v\in T$. We need to either determine their
LCE, or reduce the query to determining the LCE of two paths of length at most $\log^{2}b$.
If $\ell \leq \log^{2}b$, there is nothing to do.
Otherwise, first we check if the prefixes of length $\log^{2}b$ of both paths are different in $O(1)$ time. If so, we replace the paths
with their prefixes of such length and we are done. Otherwise, if $\ell \leq 2\log^{2}b$ we replace the paths with their suffixes of length
$\ell-\log^{2}b \leq \log^{2}b$ and we are done. The remaining case is that the prefixes of length
$\log^{2}b$ are identical and $\ell > 2\log^{2}b$. In such case, we can calculate $k$ such that
$2^{k}\log^{2}b \leq \ell-\log^{2}b < 2^{k+1}\log^{2}b$. Having such $k$, we cover the suffixes of length $\ell-\log^{2}b$ with two
(potentially overlapping) paths of length exactly $2^{k}\log^{2}b$. More formally, we create two pairs of paths:
\begin{enumerate}
\item $\LA(u,\depth(u)-2^{k}\log^{2}b) \leadsto u$ and $\LA(v,\depth(v)-2^{k}\log^{2}b) \leadsto v$,
\item $\LA(u,\depth(u)-\ell+\log^{2}b) \leadsto \LA(u,\depth(u)-\ell+\log^{2}b+2^{k}\log^{2}b)$ and $\LA(v,\depth(v)-\ell+\log^{2}b) \leadsto \LA(v,\depth(v)-\ell+\log^{2}b+2^{k}\log^{2}b)$.
\end{enumerate}
If the paths from the first pair are different, it is enough to compute their LCE. If they are are identical, it is enough to
compute the LCE of the paths from the second pair. Because we can distinguish between these two cases in $O(1)$ time,
we focus on computing the LCE of two paths of length $2^{k}\log^{2}b$ ending at some $u'$ and $v'$.
The important additional property guaranteed by how we have defined the pairs is that the paths of length $\log^{2}b$ ending at
$\LA(u',\depth(u')-2^{k}\log^{2}b)$ and $\LA(v',\depth(v')-2^{k}\log^{2}b)$ are the same. Now by the properties of the difference cover
we can calculate in $O(1)$ time $d\leq \log^{2}b$ such that the $d$-th ancestors of $u'$ and $v'$ are marked. We conceptually slide both paths up
by $d$, so that they both end at these marked nodes. Because of the additional property, either the paths of length $2^{k}\log^{2}b$ ending at
$\LA(u',\depth('u)-d)$ and $\LA(v',\depth(v')-d)$ are identical, or their first mismatch actually corresponds to the LCE
of the original paths ending at $u'$ and $v'$. These two cases can be distinguished in $O(1)$ time. Then we either use the preprocessed data
to calculate the LCE in $O(1)$ time, or we are left with the suffixes of length $d$ of the paths ending at $u'$ and $v'$.
But because $d\leq\log^{2}b$, also in the latter case we are done.
\qed
\end{proof}

We apply Lemma~\ref{lem:path-path reduction} with $b=n,\log^{2}n,\log^{2}(\log^{2}n),\ldots$ terminating when $b\leq 4$.
The total number of applications is just $O(\log^{*}n)$,
because $\log^{2}(\log^{2}z)=4\log^{2}(\log z)\leq \log z$ for $z$ large enough\footnote{This follows from
$\lim_{z\to \infty}\frac{\log^{2}(\log^{2}z)}{\log z}=\lim_{z\to \infty} \frac{4\log(\log^{2}z)}{\ln z}=\lim_{z\to \infty} \frac{8}{\ln z}=0$.}. Therefore, the total space usage becomes $O(n\log^{*}n)$ and, by iteratively applying the reduction step,
for any two paths of length at most $n$ ending at given $u$ and $v$ we can in $O(\log^{*}n)$ time either compute their LCE,
or reduce the query to computing the LCE of two paths of length $O(1)$, which can be computed naively in additional
$O(1)$ time.

To prove Theorem~\ref{thm:lcepp}, we need to decrease the space usage from $O(n\log^{*}n)$ down to $O(n)$.
To this end, we create a smaller tree $T'$ on $O(n/b)$ nodes, where $b=\log^{*}n$ is the parameter of the difference cover, as follows. 
Every marked node $u\in T$ becomes a node of $T'$. The parent of $u\in T$ in $T'$ is the node corresponding in $T'$ to the $(b^{2})$-th ancestor
of $u$ in $T$, which is always marked. Additionally, we add one artificial node, which serves as the root of the whole $T'$, and make it the parent
of all marked nodes at depth (in $T$) less than $b^{2}$. Now edges of $T'$ correspond to paths of length $b^{2}$ in $T$ (except for the
edges outgoing from the root; we will not be using them). We need to assign unique names to these paths, so that the
names of two paths are equal iff the paths are the same. This can be done by traversing the suffix tree of $T$ in $O(n)$ time.
Finally, $T'$ is preprocessed by applying Lemma~\ref{lem:path-path reduction} $O(\log^{*}n)$ times as described above. Because
its size of $T'$ is just $O(n/b)$, the total space usage preprocessing time is just $O(n)$ now.

To compute the LCE of two paths of length $\ell$ ending at $u\in T$ and $v\in T$, we first compare their
prefixes of length $b^{2}$. If they are identical, by the properties of the difference cover we can calculate $d\leq b^{2}$ such that
the $d$-th ancestors of both $u$ and $v$, denoted $u'$ and $v'$, are marked, hence exist in $T'$. Consequently, if the prefixes of length
$\ell-d$ of the paths are different, we can calculate their first mismatch by computing the first mismatch of the paths of length
$\lfloor (\ell-d)/b^{2} \rfloor$ ending at $u'\in T'$ and $v'\in T'$. This follows because every edge of $T'$ corresponds to a path of length
$b^{2}$ in $T$, so a path of length $\lfloor (\ell-d)/b^{2} \rfloor$ in $T'$ corresponds to a path of length belonging to $[\ell-d-b^{2},\ell-d]$
in $T$, and we have already verified that the first mismatch is outside of the prefix of length $b^{2}$ of the original paths.
Hence the first mismatch of the corresponding paths in $T'$ allows us to narrow down where the first mismatch of the original paths in $T$
occurs up to $b^{2}$ consecutive edges.
All in all, in $O(1)$ time plus a single path-path LCE query in $T'$ we can reduce the original query to a query concerning
two paths of length at most $b^{2}$.

The final step is to show that $T$ can be preprocessed in $O(n)$ time and space, so that the LCE of any two
paths of length at most $b^{2}$ can be calculated in $O(b)$ time. We assign unique names to all paths of length $b$ in $T$, which can be
again done by traversing the suffix tree of $T$ in $O(n)$ time. More precisely, every $u\in T$ such that $\depth(u)\geq b$ stores
a single number, which is the name of the path of length $b$ ending at $u$. To calculate the LCE of
two paths of length at most $b^{2}$ ending at $u\in T$ and $v\in T$, we proceed as follows. We traverse both paths in parallel
top-down moving by $b$ edges at once. Using the preprocessed names, we can check if the first mismatch occurs on these $b$
consecutive edges, and if so terminate. Therefore, after at most $b$ steps we are left with two paths of length at most $b$, such
that computing their LCE allows us to answer the original query. But this can be calculated by
by naively traversing both paths in parallel top-down. The total query time is $O(b)$.

To summarize, the total space and preprocessing time is $O(n)$ and the query time remains $O(\log^{*}n)$, which proves
Theorem~\ref{thm:lcepp}.

\section{Path-Tree LCE}
\label{sec:path-tree}

In this section we prove the following theorem.
\begin{theorem}
For a tree $T$ with $n$ nodes, a data structure of size $O(n)$ can be constructed in $O(n)$ time to answer path-tree LCE queries in
$O((\log\log n)^{2})$ time. 
\label{thm:lcept}
\end{theorem}

\noindent The idea is to apply the difference covers recursively with the following lemma.

\begin{lemma}
For a tree $T$ with $n$ nodes and a parameter $b$, a data structure of size $O(n)$ can be constructed in $O(n\log n)$ time, so that given a path of
length $\ell \leq b$ ending at $u\in T$ and a subtree rooted at $v\in T$ we can reduce the query  in $O(\log\log n)$ time so that
the path is of length at most $b^{4/5}$.
\label{lem:path-tree reduction}
\end{lemma}

\begin{proof}
The first part of the structure is designed so that we can detect in $O(1)$ time if the path-tree LCE is of length at most $b^{4/5}$.
We consider all paths of length exactly $b^{4/5}$ in the tree. We assign names to every such path, so that testing if two paths are identical
can be done by looking at their names. Then, for every node $w$ we gather all paths of length $b^{4/5}$ starting at $w$ (i.e.,
$w\leadsto v$, where $w=\LA(v,\depth(v)-b^{4/5})$) and store their names in a perfect hash table, where every name
is linked to the corresponding node $w$. This allows us to check if the answer is at least $b^{4/5}$ by first looking
up the name of the prefix of length $b^{4/5}$ of the path, and then querying the perfect hash table kept at $v$.
If the name does not occur there, the answer is less than $b^{4/5}$ and we are done. Otherwise, we can move by $b^{4/5}$ down, i.e.,
decrease $\ell$ by $b^{4/5}$ and replace $v$ with its descendant of distance $b^{4/5}$. 

The second part of the structure is designed to work with the marked nodes. We apply Lemma~\ref{lem:difference} with $x=b^{2/5}$
and consider \emph{canonical paths} of length $i\cdot x^{2}$ in the tree, where $i=1,2,\ldots,\sqrt{x}$, ending at marked nodes. The total
number of such paths is $O(n/\sqrt{x})$, because every marked node is the endpoint of at most $\sqrt{x}$ of them.
We lexicographically sort all canonical paths and store the longest common extension of every two neighbours on the global
sorted list augmented with a range minimum query structure. Also, for every marked node $v$ and every  $i=1,2,\ldots,x$, we save the position of the path
$\LA(v,\depth(v)-i\cdot x^{2}) \leadsto v$ on the global sorted list. 
Additionally, at every node $u$ we gather all canonical paths starting there, i.e., $u \leadsto v$ such that $\LA(v,\depth(v)-i\cdot x^{2})=u$
for some $i=1,2,\ldots,x$, sort them lexicographically and store on the local sorted list of $u$. Every such path is represented by a pair $(u,i)$.
The local sorted list is augmented with a predecessor structure storing the positions on the global sorted list.

Because we have previously decreased $\ell$ and replaced $v$, now by the properties of the difference cover we can find $d\leq x^{2}$
such that the $(\ell+d)$-th ancestor of $u$ and the $d$-th ancestor of $v$ are marked, and then increase $\ell$ by $d$ and replace
$v$ by its $d$-th ancestor. Consequently, from now on we assume that both $\LA(u,\depth(u)-\ell)$ and $v$ are marked.

Now we can use the second part of the structure. If $\ell \leq b^{4/5}$, there is nothing to do. Otherwise, the prefix of
length $\lfloor \ell / x^{2}\rfloor \cdot x^{2}$ of the path is a canonical path (because $\ell \leq \sqrt{x}\cdot x^{2}$),
so we know its position on the global sorted list.
We query the predecessor structure stored at $v$ with that position to get the lexicographical predecessor and successor of the prefix
among all canonical paths starting at $v$. This allows us to calculate the longest common extension $p$ of the prefix and all canonical paths
starting at $v$ by taking the maximum of the longest common extension of the prefix and its predecessor, and the prefix and its successor.
Now, because canonical paths are all paths of the form $i \cdot x^{2}$, the length of the path-tree LCE cannot exceed $p+x^{2}$.
Furthermore, with a level ancestor query we can find $v'$ such that the paths $\LA(u,\depth(u)-\ell) \leadsto \LA(u,\depth(u)-\ell+p)$ and
$v \leadsto v'$ are identical. Then, to answer the original query, it is enough to calculate the path-tree LCE for
$\LA(u,\depth(u)-\ell+p) \leadsto \LA(u,\depth(u)-\ell+\min(\ell,p+x^{2}))$ and the subtree rooted at $v'$.
Therefore, in $O(\log\log n)$ time we can reduce the query so that the path is of length at most $x^{2}=b^{4/5}$ as claimed.

To achieve $O(n)$ construction time, we need to assign names to all paths of length $b^{4/5}$ in the tree, which can be done in
$O(n)$ by traversing the suffix tree of $T$. We would also like to lexicographically sort all canonical paths, but this seems
difficult to achieve in $O(n)$. Therefore, we change the lexicographical order as follows: we assign names to all canonical
paths of length exactly $x^{2}$, so that different paths get different names and identical paths get identical names
(again, this can be done in $O(n)$ time by traversing the suffix tree). Then we treat every canonical path of length $i\cdot x^{2}$
as a sequence consisting of $i$ names, and sort these sequences lexicographically in $O(n)$ time with radix sort.
Even though this is not the lexicographical order, the canonical paths are only used to approximate the answer up to an additive
error of $x^{2}$, and hence such modification is still correct.
\qed
\end{proof}

We apply Lemma~\ref{lem:path-tree reduction} with $b=n,n^{4/5},n^{(4/5)^{2}},\ldots,1$. The total number of applications is
$O(\log\log n)$. Therefore, the total space usage becomes $O(n\log\log n)$, and by applying the reduction step iteratively,
for any path of length $n$ ending at $u$ and a subtree rooted at $v$ we can compute the path-tree LCE in $O((\log\log n)^{2})$ time.
The total construction time is $O(n\log\log n)$.

To prove Theorem~\ref{thm:lcept}, we need to decrease the space usage and the construction time. The idea is similar to the one
from Section~\ref{sec:path-path}: we create a smaller tree $T'$ on $O(n/b)$ nodes, where $b=\log\log n$ is the parameter
of the difference cover. The edges of $T'$ correspond to paths of length $b^{2}$ in $T$. We preprocess $T'$ as described
above, but because its size is now just $O(n/b)$, the preprocessing time and space become $O(n)$.

To compute the path-tree LCE for a given path of length $\ell$ ending at $u$ and a subtree rooted at $v$, we first check if the
answer is at least $b^{2}$. This can be done in $O(\log\log n)$ time by preprocessing all paths of length $b^{2}$ in $T$,
as done inside Lemma~\ref{lem:path-tree reduction} for paths of length $b^{4/5}$. If so, we can decrease $\ell$ and
replace $v$ with its descendant, so that both $\LA(u,\depth(u)-\ell)$ and $v$ are marked, hence exist in $T'$. Then
we use the structure constructed for $T'$ to reduce the query, so that the path is of length at most $b^{2}$. Therefore,
it is enough how to answer a query, where a path is of length at most $b^{2}$, in $O(\log\log n)$ time after $O(n)$ time
and space preprocessing.

The final step is to preprocess $T$ in $O(n)$ time and space, so that the path-tree LCE of a path of length at most $b^{2}$
and any subtree can be computed in $O(b)$ time. We assign unique names to all paths of length $b$ in $T$. Then, for every $u$
we gather the names of all paths $u\leadsto v$ of length $b$ in a perfect hash table. To calculate the path-tree LCE, we traverse
the path top-down while tracing the corresponding node in the subtree. Initially, we move by $b$ edges by using the
perfect hash tables. This allows us to proceed as long as the remaining part of the LCE is at least $b$. Then, we traverse the remaining
part consisting of at most $b$ edges naively. In total, this takes $O(b)$ time. The space is clearly $O(n)$ and the preprocessing
requires constructing the perfect hash tables, which can be done in $O(n)$ time.

\subsection{Lower Bound}

In this section, we prove that any path-tree LCE structure of size $O(n \polylog(n))$ must necessarily
use $\Omega(\log\log n)$ time to answer queries. 
As shown by P{\v{a}}tra{\c{s}}cu and Thorup~\cite{PT2006}, for $U=n^{2}$ any predecessor structure consisting
of $O(n \polylog(n))$ words needs $\Omega(\log\log n)$ time to answer queries, assuming that the word size is $\Theta(\log n)$.
We show the following reduction, which implies the aforementioned lower bound.

\begin{theorem}
For any $\epsilon>0$, given an $\LCEPT$ structure that uses $s(n)=\Omega(n)$ space and answers queries in $q(n)=\Omega(1)$ time we can build
a predecessor structure using $O(s(2U^{\epsilon}+n\log|U|))$ space and $O(q(2U^{\epsilon}+n\log|U|))$ query time for any $S\subseteq [0,U)$ of
size $n$.
\end{theorem}

\begin{proof}
We construct a tree $T$ consisting of two parts, which are then glued together by adding an artificial root. One part is simply a full binary trie
on $|U|^{\epsilon}$ leaves (for simplicity, we assume that $|U|^{\epsilon}$ is a power of $2$), each of them corresponding to an element of
$[0,U^{1/\epsilon})$. The other is a binary trie containing all elements of $S$.
More precisely, for every $x\in S$ we consider the binary expansion of $x$, which is of length $\log|U|$, and insert the corresponding path
into the trie. The resulting tree is of size $2|U|^{\epsilon}+n\log |U|$. Then we can find the predecessor of any $x$ in $S$ with $\LCEPT$ queries as
follows. First, observe that a predecessor query can be seen as starting at the root of the binary trie containing the elements of $S$ and
navigating it according to the binary expansion of $x$ as long as possible. There are three cases.

\begin{enumerate}
\item the search ends at a leaf. In such case, $x\in S$.
\item the search ends at a node $u$ such that the next binary digit of $x$ is $0$ but $u$ has no left child. Then the predecessor of $x$ is
the predecessor of the element corresponding to the leftmost leaf in the right subtree of $u$.
\item the search ends at a node $u$ such that the next binary digit of $x$ is $1$ but $u$ has no right child. Then the predecessor of $x$ corresponds
to the rightmost leaf in the left subtree of $u$.
\end{enumerate}

In the first case, we are done. In the second and third case, the answer depends only on the node $u$, so by storing an additional data of size $O(1)$
at every node of $T$ we can locate the predecessor in $O(1)$ time after having found the node $u$.

We split the binary expansion of $x$ into
$1/\epsilon$ chunks $x_{1},x_{2},\ldots,x_{1/\epsilon}$ and process them one-by-one. First, we determine the largest $i$ such that the binary
trie containing all elements of $S$ contains a node $v$ corresponding to the prefix $x_{1}x_{2}\ldots x_{i}$ of the binary expansion of $x$.
This can be done in $O(1)$ time and $O(n\log|U|)$ space using perfect hashing. If $i=1/\epsilon$, $u=v$ and we are done. Otherwise, let
$r$ be the root of the full binary trie and $\ell$ its leaf corresponding to $x_{i+1}$, which can be explicitly stored for every element of
$[0,U^{1/\epsilon})$. Then a single $\LCEPT(r,\ell,v)$ query allows us to determine the node $u$. Overall, the query takes
$O(1/\epsilon+1+q(2U^{\epsilon}+n\log|U|))=O(q(2U^{\epsilon}+n\log|U|))$ time. \qed
\end{proof}

By applying the reduction with $U=n^{2}$ and $\epsilon=1/2$, we get that an $\LCEPT$ structure using $O(n\polylog(n))$ space and answering queries
in $o(\log\log n)$ time implies a predecessor structure using $O(n\polylog(n))$ space and answering queries in $o(\log\log(n))$ time,
which is not possible.

\section{Tree-Tree LCE} 
We now consider the $\LCETT$ problem. We show that the problem is set intersection hard and give a time-space trade-off. 

\subsection{The Set Intersection Reduction}
The \emph{set intersection problem} is defined as follows. Given a family $\mathcal{S} = \{S_1, \ldots, S_k\}$ of sets of total size $n  = \sum_{i=1}^k |S_i|$ the goal is to preprocess $\mathcal{S}$ to answer queries: given two sets $S_i$ and $S_j$ determine if $S_i \cap S_j = \emptyset$. The set intersection problem is widely believed to require superlinear space in order to support fast queries. A folklore conjecture states that for sets of size polylogarithmic in $k$, supporting queries in constant time requires $\tilde \Omega(k^2)$ space~\cite{PR2014} (see also~\cite{CP2010}).

We show the following reduction. 
\begin{theorem}
Let $T$ be a tree with $n$ nodes. Given an $\LCETT$ data structure that uses $s(n)$ space and answers queries in $q(n)$ time we can build a set intersection data structure using $O(s(n))$ space and $O(q(n))$ query time, for input sets containing $O(n)$ elements.
\end{theorem}
\begin{proof}
Let $\mathcal{S} = \{S_1, \ldots, S_k\}$ be an instance of set intersection with $n = \sum_{i=1}^k |S_i|$. We transform the sets into a tree $T$ with root $r$. For each set $S_i$ create a node $v_i$ as a child of $r$. For each element $e \in S_i$ create a child node of $v_i$ and label the edge by $e$. See Fig.~\ref{fig:setintersection} for an example. To answer an intersect query for $S_i$ and $S_j$ we compute $\ell = \LCETT(v_i, v_j)$. If $\ell = 1$ then $S_i$ and $S_j$ intersect and if $\ell = 0$ they don't. \qed
\end{proof}
\begin{figure}
\begin{center}
\includegraphics[scale=0.6]{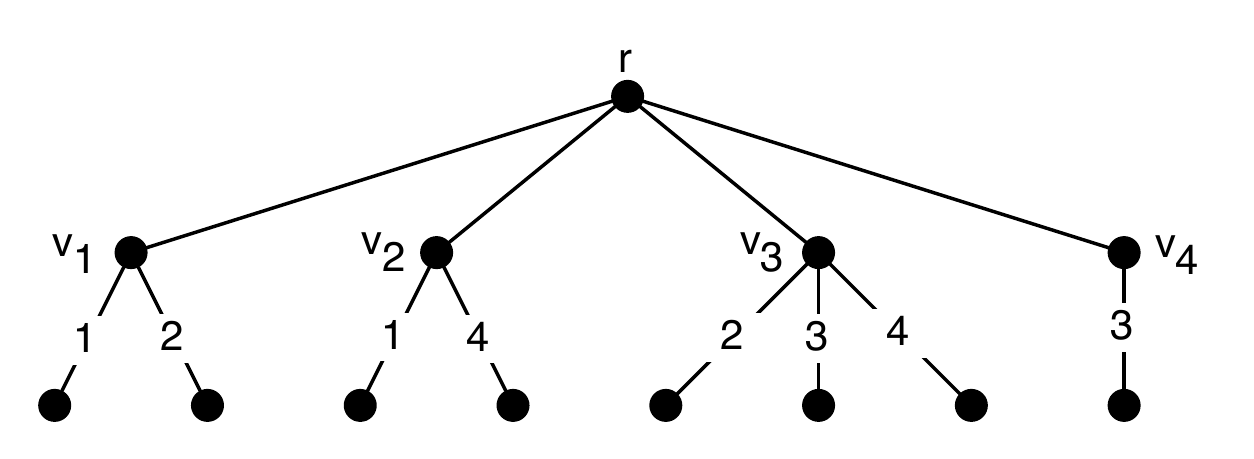}
\end{center}
\caption{Set intersection lower bound example.  $S_1=\{1,2\}, S_2=\{1,4\}, S_3=\{2,3,4\}, S_4=\{3\}$. For example, $\LCETT(v_1, v_3)$=1 since  $S_1 \cap S_3 \neq \emptyset$ but $\LCETT(v_2, v_4)$=0 since  $S_2 \cap S_4 = \emptyset$. }
\label{fig:setintersection}
\end{figure}

\subsection{The Time-Space Trade-Off} 
We now give a time-space trade-off for the $\LCETT$ problem as stated by the following theorem. 

\begin{theorem}\label{thm:lcett}
For a tree $T$ with $n$ nodes and a parameter $\tau$, $1 \leq \tau \leq n$, a data structure of size $O(n\tau)$ can be constructed in $O(n\tau)$ time to answer tree-tree LCE queries in $O(n/\tau)$ time. 
\end{theorem}

First consider the following two extreme solutions. Given nodes $v_1$ and $v_2$ we can simply traverse the entire subtrees $T(v_1)$ and $T(v_2)$ in parallel and report the maximal path-path LCE. Since we only need to store $T$, this solution uses $O(n)$ space and $O(|T(v_1)| + |T(v_2)|) = O(n)$ query time. On the other hand, if we preprocess and store the maximal tree-tree LCE for every pair of nodes we use $O(n^2)$ space and support queries in $O(1)$ time. We show how to efficiently balance between these solutions by clustering $T$ into $O(\tau)$ overlapping subtrees of size $O(n/\tau)$.

\paragraph{Clustering.} 
Let $C$ be a connected subgraph of $T$. A node in $V(C)$ adjacent to a node in $V(T)\backslash V(C)$ is called a \emph{boundary node} of $C$. A \emph{cluster} of $T$ is a connected subgraph of $T$ with at most two boundary nodes and at least $1$ edge. A set of clusters $CS$ is a \emph{cluster partition} of $T$ iff $V(T) = \cup_{C \in CS} V(C)$, $E(T) = \cup_{C\in CS} E(C)$, and for any $C_1, C_2 \in CS$, $E(C_1) \cap E(C_2) = \emptyset$. We will use the following clustering results which follows from Frederickson~\cite{frederickson1997} (see also~\cite{AHLT1997, AHT2000, BG2011}).

\begin{lemma}\label{lem:clustering}
Given a tree $T$ with $n>1$ nodes and a parameter $\tau$, we can construct a cluster partition $CS$ in $O(n)$ time, such that $|CS| = O(\tau)$ and $|V(C)| = O({n/\tau})$ for any $C \in CS$. 
\end{lemma}

\paragraph{The data structure.}
Our data structure consists of the following parts: 
\begin{itemize} 
\item A cluster partition $CS$ of $T$ with parameter $\tau$.
\item For each pair $(v,b)$, where $v$ is a node in $T$ and $b$ is a boundary node, we store $\LCETT(v,b)$. By Lemma~\ref{lem:clustering}, the total number of boundary nodes is $O(\tau)$ hence this uses $O(n\tau)$ space. 
\end{itemize}

\paragraph{Answering queries.}
Let $v_1$ and $v_2$ be nodes in clusters $C_1$ and $C_2$, respectively. We compute  $\LCETT(v_1, v_2)$ as follows. If $v_1$ or $v_2$ is a boundary node we return the precomputed stored answer in $O(1)$ time. Otherwise, we traverse  in parallel the part of subtree $T(v_1)$ inside $C_1$ and the part of subtree $T(v_2)$ inside $C_2$. If either endpoint of the traversal reaches a boundary node we lookup the precomputed answer. The corresponding $\LCETT$ is then the distance to the endpoint plus the precomputed answer. The answer to $\LCETT(v_1, v_2)$ is the maximal path-path LCE found during the traversal.

Since $C_1$ and $C_2$ contain $O(n/\tau)$ nodes the total time is $O(n/\tau)$. Hence, our solution uses $O(n\tau)$ space and $O(n/\tau)$ query time, thus completing the proof of Theorem~\ref{thm:lcett}.

\medskip 

\noindent Finally, we note that the above data structure can be easily modified (while maintaining the same time and space bounds) to support the following tree-tree LCE query:
Given nodes $v_1$ and $v_2$, find the largest common subtree (rather than subpath) starting at $v_1$ and $v_2$ (i.e., the largest connected subgraph that includes $v_1$ and  descendants of $v_1$ and is equal to a connected subgraph that includes $v_2$ and  descendants of $v_2$).

\bibliographystyle{abbrv}
\bibliography{paper}

\end{document}